\documentclass[preprint,12pt,authoryear]{elsarticle}
\usepackage{graphicx}
\usepackage{amsmath,amssymb}
\usepackage{multirow}
\usepackage{amsbsy}
\usepackage{color}
\usepackage{multirow}
\usepackage{amsthm}
\usepackage{framed}
\usepackage[a4paper,width=175mm,top=20mm,bottom=20mm]{geometry}

\newtheorem{theorem}{Theorem}[section]
\newtheorem{lemma}[theorem]{Lemma}
\newtheorem{corollary}[theorem]{Corollary}
\newtheorem{remark}{Remark}
\newtheorem{exmp}{Example}[section]

\numberwithin{equation}{section}
\numberwithin{table}{section}
\everymath=\expandafter{\the\everymath\displaystyle}

\begin{document}

\begin{frontmatter}

\title{Optimal two-level designs for partial profile choice experiments}

\author{Soumen Manna, Ashish Das}

\address{Department of Mathematics, Indian Institute of Technology Bombay, India.}

\begin{abstract}
We improve the existing results of optimal partial profile paired choice designs and provide new designs for situations where the choice set sizes are greater than two.  The optimal designs are obtained under the main effects models and the broader main effects model.  
\end{abstract}


\end{frontmatter}

\section{Introduction}
Discrete choice experiments are a standard tool in marketing research for quantifying
consumer preferences. In a discrete choice experiment, levels of different attributes are considered jointly in a product profile or alternative. A choice experiment uses a selected number of profiles which are grouped into choice sets. Such a group of choice sets is called a choice design. Consider a choice design consisting of $N$ choice sets, each containing $m$ profiles and each profile is a combination of $n$ factors each at two levels. Each choice set represents a virtual market from which respondents indicate the product that they prefer. The statistical analysis of the respondents' choices employs discrete choice modeling to estimate the effects of attributes or their interactions. 

As the number of factors $n$ increases, respondents often feel the cognitive burden and thus tend to exhibit inconsistent choice behavior. This is termed as information overload in marketing literature. To overcome this situation, researchers propose partial profile choice experiments where a certain number of factors remain unchanged in all the profiles of each choice set. The profile attributes that vary in a partial profile choice set are called active factors for that choice set. These active factors may differ from choice set to choice set but the number of active factors in each choice set always remains same. The number of active factors in a choice design is called the profile strength and is denoted by $\rho$. Clearly, when $\rho = n$, a partial profile choice experiment is same as a full profile choice experiment.  
%

\cite{r4} present a review of partial profile optimal choice designs. The literature so far on partial profile choice designs is mainly focused on paired choice designs under the main effects model (where two and higher order interaction effects are assumed to be zero). In this paper, under the main effects model, we provide a method of construction leading to optimal designs requiring lesser number of choice sets than the existing optimal partial profile paired choice designs. We also extend the results to the general choice sets of size $m$. Finally, we obtain optimal designs under the broader main effects model (where two factor interactions effects are not assumed to be zero but three and higher order interaction effects are assumed to be zero).

Under the multinomial logit model, in Section 2, we give the general set-up and characterize the information matrix. Section 3 provides optimal designs under the main effects model while Section 4 gives optimal designs under the broader main effects model.

\section{General set-up and the information matrix }
Consider a choice experiment with $n$ factors each at two levels, 0 and 1. We denote a choice set of size $m$ by $(T_1, T_2, \ldots, T_m)$, where a typical treatment combination $T_{i}=(i_{1} \ldots i_{r} \ldots i_{n}),i_{r}=0,1; r = 1,2,\ldots,n$. Then a partial profile choice design $d = d(N,n,m,\rho)$ is defined as $d = \{ (T_1, T_2, \ldots, T_m): 1_w = \cdots =m_w, \text{for at least $(n-\rho)$ components}$ $ w \in (1,2,\ldots,n)\}$, i.e. $d$ is a choice design such that for each choice set $ (T_1, T_2, \ldots, T_m)$, there are at least $(n-\rho)$ factors that remain in constant level among all the treatments $ T_1, T_2, \ldots, T_m$. Clearly, for $m=2$, this definition reduces to the one given by \cite{r3}.

Let $A_i$, $i = 1,2,\ldots, m$, be $N\times n$ matrices with entries 0 and 1. Then a partial profile choice design $d$ can also be represented in matrix notation as $d=(A_1, A_2, \ldots, A_m)$, where the $p$-th row from each $A_i$ makes the $p$-th choice set  $S_{pm}$ (say) and hence $d = \{ S_{pm}: p = 1,2,\ldots,N\}$. Henceforth in this paper, by a design or a choice design, we mean a partial profile choice design. 

Let $\Lambda$ be the information matrix for treatment effects corresponding to a choice design $d$, and let $B$ be the orthonormal treatment contrast matrix corresponding to the factorial effects of interest. Then the information matrix of $d$ for the factorial effects of interest is $C_d = B\Lambda B'$ (see, \cite{r7} for details). A design $d$ is connected if the corresponding information matrix $C_d$ ($=C$ say) is positive definite. A  connected design allows the estimation of all underlying factorial effects of interest. Let $\mathcal{D}_{N,n,m, \rho}$ be the class of all connected designs with $N$ choice sets each of size $m$ with $n$ factors and $\rho$ strength. 

We now characterize the $C$-matrix which helps us for the further development of this paper. Note from \cite{r7} that for a design $d \in \mathcal{D}_{N,n,m, \rho} $, the $2^n\times 2^n$ information matrix $\Lambda = ((\lambda_{st}))$ of the treatment effects for equally attractive options is 

\begin{equation*}\label{}
\lambda_{st} = \left\{ \begin{array}{ccl}
((m-1)/N m^2) \sum_{j_2 < j_3 < \cdots < j_m} N_{j_1 j_2 \ldots j_m} & \text{if} & s = t = j_1 \\

 (-1/N m^2) \sum_{j_3 < j_4 < \cdots < j_m} N_{j_1 j_2 \ldots j_m} &  \text{if} &  s =j_1, t = j_2 \\

0 &   \text{otherwise}, &
\end{array}\right.
\end{equation*}
where, $N_{j_1 j_2 \ldots j_m}$ is the indicator function taking value 1 if $(T_{j_1}, T_{j_2},\ldots, T_{j_m}) \in d$ and 0 otherwise.

Let $M^{(j_{1}j_{2}\ldots j_{m})} = ((m_{st}))$ be a $2^n\times 2^n$ matrix corresponding to a choice set $(T_{j_{1}},$ $T_{j_{2}},$ $\ldots, T_{j_{m}})$, where,
\vspace{-.3cm}
$$ m_{st} = \left\{ \begin{array}{ll}
m-1 & \text{if} \hspace{.5cm} s = t, \;t \in \{ j_{1},j_{2},\ldots, j_{m} \} \vspace{-.3cm} \\
-1 & \text{if} \hspace{.5cm} s\neq t, (s,t) \in \{j_{1},j_{2},\ldots, j_{m}\} \vspace{-.3cm} \\
0 & \text{otherwise.}
\end{array}\right.$$
Then for any choice design $d$, we can write
\begin{equation*} 
\Lambda = (1/N m^2)\sum_{j_{1}<j_{2}< \cdots <j_{m}} N_{j_1 j_2 \ldots j_m}M^{(j_1 j_2 \ldots j_m)} = (1/N m^2) \Lambda^* \hspace{.1cm} (say).
\end{equation*}
We consider the matrix $M^{(j_{1}j_{2}\ldots j_{m})}$ as the contribution of the choice set $(T_{j_{1}},$ $T_{j_{2}},$ $\ldots, T_{j_{m}})$ to $\Lambda$.
The definition of $M^{(j_{1}j_{2}\ldots j_{m})}$ suggests that we can write
$M^{(j_{1}j_{2}\ldots j_{m})} = \sum_{j_{r}<j_{r'}}M^{(j_{r}j_{r'})}$,
where, $j_{r},j_{r'} \in \{j_{1},j_{2},\ldots ,j_{m}\}$. 
Thus, the contribution of the choice set $(T_{j_{1}},$ $T_{j_{2}},$ $\ldots, T_{j_{m}})$ to $\Lambda$ is the sum of the contributions of all its $m(m-1)/2$ component pairs $(T_{j_{r}}, T_{j_{r'}})$. Therefore, the information matrix for the factorial effects of interest is  
\begin{eqnarray}\label{BLB2}
C = B\Lambda B'
&= (1/N m^2) \sum_{j_{1}<j_{2}< \cdots <j_{m}} N_{j_1 j_2 \ldots j_m}\left\{B \left(\sum_{j_{r}<j_{r'}}M^{(j_{r}j_{r'})}\right) B'\right\}.
\end{eqnarray}
\noindent Let $f_1,f_2,\ldots,f_n$, denote the $n$ factors. Let $\theta_1 = (F_1, F_2, \ldots, F_n)'$ denotes $n$ main effects and $ \theta_2 =  (F_{12},F_{13},\ldots, F_{(n-1)n})'$ denotes $n(n-1)/2$ two factor interaction effects. We define the positional value of the  $h$-th factor $f_h$ in $T_i$ as $i_h$. Corresponding to a pair $(T_i, T_j)$, we define the positional value of $F_h$ and $F_{k}$ in $(T_i, T_j)$ as $(i_hi_k,j_hj_k)_{hk}$. Similarly, the positional value of  $F_h$ and $F_{hk}$ in $(T_i, T_j)$ as $(i_hi_k,j_hj_k)_{(hk)}$ and the positional value of  $F_h$ and $F_{kl}$ in $(T_i, T_j)$ as $(i_h(i_k i_l), j_h(j_k j_l))_{h(kl)}$.

\begin{lemma}\label{lm:C_hkVal-3}
	Let $B_{x} = (x_{1},\ldots, x_{i},\ldots, x_{j},\ldots,x_{2^n})$ and $B_{y} = (y_{1},\ldots, y_{i},\ldots, $ $y_{j},\ldots, y_{2^n})$ be any two contrast vectors with $2^n$ elements. Then for a given {\it component pair} $(T_i,T_j)$, the value of $B_{x}M^{(ij)}B'_{y}=(x_i-x_j)(y_i -y_j)$.
\end{lemma}

\begin{proof}
	Note that $M^{(ij)}$ is a $ 2^n\times 2^n$ matrix with all elements $0$ except $ M_{ii}^{(ij)} =  M_{jj}^{(ij)} = 1$ and $ M_{ij}^{(ij)} =  M_{ji}^{(ij)} = -1$. Therefore, $B_{x}M^{(ij)}B'_{y} = \left(0,\ldots,(x_{i}-x_{j}),\ldots, -(x_{i}-x_{j}) ,\ldots,0\right)B'_{y}$ \\
	$= (x_{i}-x_{j})y_{i} -(x_{i}-x_{j})y_j = (x_{i}-x_{j})(y_{i}-y_j)$.
\end{proof}

Let $B_{(1)}=((b_{hi}))$ be the orthogonal contrast matrix of $\theta_1$. Corresponding to a treatment 
$T_i$ and factorial effect $F_h$, let $b_{hi}=-1$ if $i_h=0$ and $b_{hi}=1$ 
if $i_h=1$. Let $B_{(2)}=((b_{kli}))$ be the orthogonal contrast matrix of $\theta_2$. Corresponding to a 
treatment $T_i$ and factorial effect $F_{kl}$, let $b_{kli} = b_{ki}b_{li}$. It is assumed that the 
treatments are arranged in lexicographic order in $B_{(1)}$ and $B_{(2)}$. Therefore, the contrast vector corresponding to the main effect $F_h$ of $B_{(1)}$ is $B_{h} = \otimes_{i =1}^{h-1} 
\left(1 \hspace{0.2cm} 1\right) \otimes \left(-1 \hspace{0.2cm} 1\right)
\otimes_{i =h+1}^{n} \left(1 \hspace{0.2cm} 1\right),$ 
and the contrast vector corresponding to the two factor interaction effects $F_{kl}$ of $B_{(2)}$ is
$ B_{kl} = \otimes_{i =1}^{k-1} \left(1 \hspace{0.2cm} 1\right) \otimes \left(-1 \hspace{0.2cm} 1\right)
\otimes_{i =k+1}^{l-1} \left(1 \hspace{0.2cm} 1\right) \otimes \left(-1 \hspace{0.2cm} 1\right)
\otimes_{i =l+1}^{n} \left(1 \hspace{0.2cm} 1\right)
$. Here, $\otimes$ denotes kronecker product. Then, we have the following result, proof of which follows simply by definitions.  

\begin{lemma}\label{lemma:converse-3}
	For $h \neq k \neq l,(h,k,l) \in \{1,\ldots , n\}$ and corresponding to the component pair $(T_i, T_j)$, the exhaustive cases indicating possible values of $B_{h}M^{(ij)}B'_{k}$, $B_{h}M^{(ij)}B'_{hk}$ and $B_{h}M^{(ij)}B'_{kl}$  are\\
	
	\noindent Case 1: For $F_h$ and $F_k$, \\
	a) $B_{h}M^{(ij)}B'_{k}=-4$ when
	$(i_{h}i_{k},j_{h}j_{k})_{hk} \equiv (01,10)_{hk}$\\
	b) $B_{h}M^{(ij)}B'_{k}=4$ when $(i_{h}i_{k},j_{h}j_{k})_{hk} \equiv (00,11)_{hk}$\\
	c) $B_{h}M^{(ij)}B'_{k}=0$ for all other situations. \\
	
	\noindent Case 2:  For $F_h$ and $F_{hk}$,\\
	a) $B_{h}M^{(ij)}B'_{hk}=4$ when $(i_{h}i_{k},j_{h}j_{k})_{(hk)} \equiv (01,11)_{(hk)}$\\
	b) $B_{h}M^{(ij)}B'_{hk}=-4$ when
	$(i_{h}i_{k},j_{h}j_{k})_{(hk)} \equiv (00,10)_{(hk)}$\\
	c) $B_{h}M^{(ij)}B'_{hk}=0$ for all other situations. \\
	
	\noindent Case 3:  For $F_h$ and $F_{kl}$,  \\
	a) $B_{h}M^{(ij)}B'_{kl}=4$ when \\
	$(i_{h}(i_{k}i_{l}), j_{h}(j_{k}j_{l}))_{h(kl)} \in \{(0(10),1(00)), (0(10),1(11)), (0(01),1(00)), (0(01),1(11))\} _{h(kl)}$\\
	b) $B_{h}M^{(ij)}B'_{kl}=-4$ when \\
	$(i_{h}(i_{k}i_{l}), j_{h}(j_{k}j_{l}))_{h(kl)} \in \{(0(00),1(10)), (0(00),1(01)), (0(11),1(10)), (0(11),1(01))\} _{h(kl)}$\\
	c) $B_{h}M^{(ij)}B'_{kl}=0$ for all other situations.
\end{lemma}


Note that each choice set $S_{pm}$ of a design $d \in \mathcal{D}_{N,n,m,\rho}$, contains $m(m-1)/2$ component pairs $(T_i,T_j)$ and there are $N$ such choice sets in $d$. Therefore, total number of component pairs in a design $d$ is $N^* = N m(m-1)/2$.

We use the universal optimality criteria for finding optimal designs in $\mathcal{D}$. Following Kiefer (1975), a choice design $d^*$ is universally optimal in $\mathcal{D}$ if $C_{d^*}$ is scalar multiple of identity matrix and $ trace(C_{d^*}) \geq trace(C_{d})$ for any other $d \in \mathcal{D}$. If a design $d$ is universally optimal, then it is also $A$-, $D$-, and $E$-optimal. Henceforth in this paper, by optimal design, we mean universally optimal design. 


\section{Optimal designs under the main effects model}
In this section, we discuss optimal designs for main effects of interest under the main effects model. We first obtain optimal designs for choice sets of size $m=2$, and then develop optimal designs for general choice sets of size $m$. The information matrix for estimating main effects of interest corresponding to a design $d \in \mathcal{D}_{N,n,m,\rho}$, is $C =  (1/2^n) B_{(1)}\Lambda B'_{(1)} = (1/2^n N m^2) B_{(1)}\Lambda^* B'_{(1)}$. For the purpose of this section, we define two more notations here. Let $\eta_{hk}^{1+}$ and $\eta_{hk}^{1-}$ be the total number of component pairs of the type $(00,11)_{hk}$ and $(01,10)_{hk}$ respectively for a design $d \in \mathcal{D}$.  

\begin{lemma}\label{lem1}
	For $h\neq k, (h,k)\in \{1,\ldots , n\}$, the $(h,k)$-th element of $C$ matrix is zero if and only if $\eta_{hk}^{1+} = \eta_{hk}^{1-}$.
\end{lemma}
\begin{proof}
	Let $c_{hk}$ and $c'_{hk}$ denote the $(h,k)$-th element of $C$ and $ B_{(1)}\Lambda^* B'_{(1)}$ respectively. Then it follows from {\it Case 1} of Lemma \ref{lemma:converse-3} that
	\begin{eqnarray*}\label{Cstarxx}
		c'_{hk} &=& \sum_{j_{1}<j_{2}< \cdots <j_{m}} N_{j_1 j_2 \ldots j_m}\sum_{j_{r}<j_{r'}}\{B_hM^{(j_{r}j_{r'})} B_{k}'\}\\
		&=& \left[4(\eta_{hk}^{1+} - \eta_{hk}^{1-})+0\{N^* - (\eta_{hk}^{1+} + \eta_{hk}^{1-}) \}\right].
	\end{eqnarray*}
	Thus $c'_{hk}$ or $c_{hk}=0$ if and only if $\eta_{hk}^{1+} = \eta_{hk}^{1-}$.
\end{proof}
We now find the expression of maximum $trace(C)$ for a design $d$ in $\mathcal{D}_{N,n,m,\rho}$.  

\begin{lemma}\label{CtrOpti_m}
	Let $n_{ph}$ = total number of 0's corresponding to the $h$-th positional value of all treatments in the choice set $S_{pm}$, $p=1,2,\ldots,N$, $h=1,2,\ldots,n$. Let $d$ be a design in  $\mathcal{D}_{N,n,m,\rho}$, then,
	$$ max(trace(C)) = \left\{ \begin{array}{cc}
	\rho / 2^n & \text{for $m$ even} \vspace{.2cm}\\
	\rho(m^2-1)/(2^nm^2) & \text{for $m$ odd,} \\
	\end{array}\right.$$
	and the maximum $trace (C)$ occurs when $n_{ph} = m/2$ ($m$ even) and 	$ n_{ph} = (m-1)/2$  or  $(m+1)/2$ ($m$ odd), for every active factor $f_h$ of every choice set $S_{pm}$.
\end{lemma}

\begin{proof}
	Let $ c'_{hh}$ be the $(h,h)$-th element of  $B_{(1)}\Lambda^* B'_{(1)}$. Every component pair $(T_i,T_j)$ adds a value $B_h M^{(ij)}B_{h}'$ at $c'_{hh}$. From Lemma \ref{lm:C_hkVal-3}, it follows that this value is 4 if and only if $i_h \neq j_h$. Note from (\ref{BLB2}) that the contribution of a choice set $S_{pm}$ to $c'_{hh}$ is equal to the sum of contributions of all its $m(m-1)/2$ component pairs. Now, if $f_h$ is an active factor in $S_{pm}$, then there are $n_{ph}(m-n_{ph})$ component pairs $(T_i,T_j)$ in $S_{pm}$ for which $i_h \neq j_h$. Thus every choice set $S_{pm}$ adds a value $4n_{ph}(m-n_{ph})$ to $c'_{hh}$ and this value is maximum 
	when (i) $n_{ph}=m/2$ (for $m$ even) and (ii) $n_{ph}=(m-1)/2$ or $n_{ph}=(m+1)/2$ (for $m$ odd). Considering the fact that there are $\rho$ active factors in each choice set $S_{pm}$ and there are $N$ such choice sets in $d$, we have the required expression for  $max(trace(C))$. 
\end{proof}

From Lemma \ref{lem1} and Lemma \ref{CtrOpti_m}, it follows that under the main effects model a design $d$ is  optimal in  $\mathcal{D}_{N,n,m,\rho}$, if (a) $C$ is diagonal, i.e., $\eta_{hk}^{1+} = \eta_{hk}^{1-}$, for all $ h\neq k, h,k \in \{1,2,\ldots,n\}$ and (b) $ trace(C)$ is maximum, i.e.,  $n_{ph} = m/2$ ( $m$ even) and 	$n_{ph} = (m-1)/2$  or  $(m+1)/2$ ($m$ odd), for every active factor position $h$ and for every choice set $S_{pm}$,   $p=1,2,\ldots,N$, $h=1,2,\ldots,n$. \\

We first discuss optimal designs for paired choice experiments and then extend those results for choice sets of size $m$. Let $X$ be the difference matrix of a paired choice design $d = (A_1, A_2)$, i.e., 
$ X = A_1 - A_2 $. Therefore, the entries of $X$ are $0, 1$ and $-1$. Note that if we know $X$, we can easily construct a design $d = (A_1,A_2)$ and vice-versa. Following \cite{r3}, \cite{r4}, \cite{r6}, we have $C = (1/N2^n) X'X$. Thus, it follows that for a connected design $N \geq n$. A design $d \in \mathcal{D}_{N, n,2, \rho} $ is said to be saturated or tight if $N = n$. We have the following lemma for optimal designs.

\begin{lemma}\label{optimalG}
	A design $d$ is optimal in $\mathcal{D}_{ N,n,2, \rho}$, if $X'X = N\rho/n I_n$. 
\end{lemma}

Note that a large number of choice sets can create cognitive burden to the respondents. Therefore, a practical number of choice sets are required for an optimal choice design. \cite{r8} have suggested possible use of weighing matrices to construct good choice designs. Later \cite{r31} have used weighing matrices for the construction of partial profile paired choice designs when there are two groups of factors.

A weighing matrix $W = W(n,\rho)$ of order $n$ and weight $\rho$ is an $n \times n $ orthogonal matrix, with entries $0$, $+1$ and $-1$ such that there are $\rho$ non-zero entries in each row and each column, i.e., $ WW' = W'W = \rho I_n$. 

 We now have the following result for saturated optimal design $d$ in $\mathcal{D}_{n, n,2,\rho}$.  
\begin{theorem}\label{SaturateTh}
	If  there exists a weighing matrix $W(n,\rho)$, then there exists an optimal design $d$ in $\mathcal{D}_{n, n,2,\rho}$.
\end{theorem}

\begin{proof}
	Let $d$ be a design with $X = W$. Therefore, 
	$ X'X = W'W = \rho I_n$. Hence, $d$ is optimal in $\mathcal{D}_{n, n,2,\rho}$ by Lemma \ref{optimalG}. 
\end{proof}

\begin{exmp}
	Let $N=n = 2^3 = 8$ and $\rho = 5$, then the difference matrix and corresponding  optimal design $d$ in $\mathcal{D}_{8,8,2,5}$ is 
	\vspace{-.7cm}
	\begin{center}
		$$ X =  \left(\begin{array}{c}
		11111000   \vspace{-.3cm} \\ 
		1u1u0100  \vspace{-.3cm} \\
		11uu0010 \vspace{-.3cm} \\
		1uu10001 \vspace{-.3cm} \\
		u0001111 \vspace{-.3cm} \\
		0u001u1u \vspace{-.3cm} \\
		00u011uu \vspace{-.3cm} \\
		000u1uu1 \vspace{-.3cm} \\
		\end{array} \right) \hspace{.2cm}\text{and} \hspace{.2cm} 
		d = \begin{array}{cc}
		(11111\!*\!**, & 00000\!*\!**) \vspace{-.3cm} \\
		(1010\!*\!1\!*\!*, & 0101\!*\!0\!*\!*) \vspace{-.3cm} \\
		(1100\!*\!*1*, & 0011\!*\!*0*) \vspace{-.3cm} \\
		(1001\!*\!*\!*\!1, & 0110\!*\!*\!*\!0) \vspace{-.3cm} \\
		(0\!*\!*\!*\!1111, & 1\!*\!*\!*\!0000) \vspace{-.3cm} \\
		(*0\!*\!*1010, & *1\!*\!*0101) \vspace{-.3cm} \\
		(*\!*\!0\!*\!1100, & *\!*\!1\!*\!0011) \vspace{-.3cm} \\
		(*\!*\!*01001, & *\!*\!*10110) \vspace{-.3cm} \\
		\end{array}. $$
	\end{center} 
	Here by `` $u$'', we denote  $-1$ in $X$ and $*$ represents the fixed positions.
\end{exmp}

Note that the existence of a $ W(n,\rho)$ ensures a optimal  design in $ \mathcal{D}_{n,n,2,\rho}$. A  good source of existing weighing matrices are given in \cite{r1}, \cite{r2}. We now focus on the cases when $W(n,\rho)$ do not exist for given $n$ and $\rho$. In this case one can increase the number of choice sets $N$ ($>n$) to find an optimal design $d$. A general construction in this regards is given by \cite{r3} using incomplete block design and Hadamard matrix for any given $n$ and $\rho$. They first construct an incomplete block design with $n$ `treatments', $n/gcd(n,\rho)$ `blocks', `block size' $\rho$, and $\rho/gcd(n,\rho)$ `replications', where ``gcd" denotes the greatest common divisor. Let $M$ be the incident matrix for such a block design. Let $H$ be the Hadamard matrix of order $h(\rho)$, where $h(\rho)$ is the least number greater than or equal to $\rho$ such that a Hadamard matrix of order $h(\rho)$ exists.  Now putting different Hadamard columns of $H$ in place of each 1 and zero vector of order $h(\rho)$ in place of each 0 corresponding to the each row of $M$ generates a difference matrix $X$ (and corresponding design $d$) in  $N = n h(\rho)/gcd(n,\rho)$ choice sets. We call this construction method as {\it Method-H}. \\

We now provide a different construction using weighing matrices instead of Hadamard matrices. Our construction improves the existing results in some cases. Suppose there exists a weighing matrix $W(\nu, \rho)$, $n\geq \nu$. Now create an $N_0 \times n$ matrix $M^*$ with entries 0 and 1 such that each row sum is $\nu$ and each column sum is equal. We can construct such matrix $M^*$ in $N_0 = n/gcd (n,\nu)$ rows as follows. Take $\nu$ consecutive 1 from the first position of first row and rests are 0. Then take $\nu$ consecutive 1 from $(\nu+1)$-th position of second row and rests are 0 and so on. Now putting $\nu$ different columns of $W$ in place of each 1 and a zero vector of order $\nu$ in place of each 0 corresponding to the each row of $M^*$ gives a difference matrix $X$ in  $N = n \nu/gcd(n,\nu)$ choice sets such that $X'X = N \rho/n I_n$. Therefore, by Lemma \ref{optimalG}, the corresponding design $d$ is optimal in $\mathcal{D}_{N,n,2,\rho}$. We call this construction method as {\it Method-W}. Hence, we have the following result.      

\begin{theorem}\label{OptimalD} 
	If there exists a weighing matrix $W(\nu,\rho)$, then there exists a optimal design $d$ in $\mathcal{D}_{N,n,2,\rho}$, for every $n \geq \nu$ with $N = n \nu/gcd (n,\nu)$.  
\end{theorem}

\begin{exmp}
	Suppose we want to construct an optimal design for $n = 10$ and $\rho = 3$. There exists two weighing matrices $ W(4,3)$ and $ W(8,3)$ for $\nu \leq 10$, $\rho= 3$. If we use $W(4,3)=(W_1,W_2,W_3,W_4)$ for construction, we get an optimal design $d_1$ in $N_1 = (10 \times 4)/gcd(10,4) = 20$ choice sets. The $M^*$ and $X$ matrix ($X_1$, say) can be written as  \vspace{-.3cm}
	$$ 
	M^* = \left( \begin{array}{cccccccccc} 
	1 & 1 & 1 & 1 & 0 & 0 & 0 & 0 & 0 & 0 \vspace{-.3cm} \\
	0 & 0 & 0 & 0 & 1 & 1 & 1 & 1 & 0 & 0 \vspace{-.3cm} \\
	1 & 1 & 0 & 0 & 0 & 0 & 0 & 0 & 1 & 1 \vspace{-.3cm} \\
	0 & 0 & 1 & 1 & 1 & 1 & 0 & 0 & 0 & 0 \vspace{-.3cm} \\
	0 & 0 & 0 & 0 & 0 & 0 & 1 & 1 & 1 & 1  \\
	\end{array} \right), $$ \vspace{-.3cm}
	$$            
	X_1 = \left( \begin{array}{cccccccccc} 
	W_1 & W_2 & W_4 & W_3 & \mathbf{0} & \mathbf{0} & \mathbf{0} & \mathbf{0} & \mathbf{0} & \mathbf{0} \vspace{-.3cm} \\
	\mathbf{0} & \mathbf{0} & \mathbf{0} & \mathbf{0} & W_4 & W_1 & W_2 & W_3 & \mathbf{0} & \mathbf{0} \vspace{-.3cm} \\
	W_3 & W_4 & \mathbf{0} & \mathbf{0} & \mathbf{0} & \mathbf{0} & \mathbf{0} & \mathbf{0} & W_2 & W_1 \vspace{-.3cm} \\
	\mathbf{0} & \mathbf{0} & W_2 & W_3 & W_4 & W_1 & \mathbf{0} & \mathbf{0} & \mathbf{0} & \mathbf{0} \vspace{-.3cm} \\
	\mathbf{0} & \mathbf{0} & \mathbf{0} & \mathbf{0} & \mathbf{0} & \mathbf{0} & W_3 & W_4 & W_2 & W_1  \\
	\end{array} \right).  \hspace { 0.3cm}
	$$           
	
	\noindent Here $X_1' X_1 = 6I_{10}$ and using Theorem \ref{OptimalD}, the corresponding design $d_1$ is optimal in $\mathcal{D}_{20,10,2,3}$.
	
	Similarly, if we use $W(8,3)$, we get an optimal design $d_2$ in $N_2 = (10 \times 8)/gcd(10,8) = 40$ choice sets in $\mathcal{D}_{40,10,2,3}$. 
	Note that if we use {\it Method-H} instead, then $M$ would be a $10\times 10$ matrix of 0 and 1, where each row sum is 3. In that case, by replacing any three columns of Hadamard matrix of order 4 in the places of 1 and a column of four zeros in places of 0 for each rows of $M$, we get a difference matrix $X$ ($X_3$, say) in $N_3 = (10 \times 4)/gcd(10,3) = 40$ choice sets where  $X_3' X_3 = 12I_{10}$. Thus the corresponding design $d_3$ is also optimal in $\mathcal{D}_{40,10,2,3}$. 
\end{exmp}


We see from the above example that the use of {\it Method-W} instead of {\it Method-H} sometimes help us to get optimal partial profile paired choice designs in lesser number of choice sets. Given $\rho$ and $n$, in order to obtain an optimal design with minimum number of choice sets, we follow the steps as below.
{\it	
\begin{enumerate}
		\item If $ W = W(n,\rho)$ exists, then an optimal design exists in $N = n$ choice sets;  otherwise go to Step 2.  
		\item Find all the $n^*_j (< n)$ such that $W(n^*_j, \rho)$ exists.
		\item Calculate $ K_j = (n n^*_j/gcd(n,n^*_j))$, for every $j$.
		\item $N_1 = min\{K_j \}$ with corresponding $n^*_j = \nu$ (say).
		\item Calculate $ N_2 = (n h(\rho)/gcd(n,\rho))$.
		\item $N = min\{N_1, N_2\}$. 
		\begin{itemize}
			\item [i)] If $N = N_1$, construct optimal design using {\it Method-W}.
			\item [ii)] If $N = N_2$, construct optimal design using {\it Method-H}.
			\item [iii)] If $N = N_1 = N_2$, construct optimal design using either {\it Method-W} or {\it Method-H}.
		\end{itemize}
	\end{enumerate}}

Table \ref{tb1} provides the minimum number of choice sets ($N$) required  for an optimal design, where, $2\leq \rho \leq 6$, $2 < n \leq 15$. 
\begin{table}[ht]
	\caption{ Given $\rho$ and $n$, required ($N$) for optimal design. }
	\centering
	\begin{tabular}{|c|c|c|c|c|c|c|c|c|c|c|c|c|c|}
		\hline 
		$\rho \diagdown n $ & 3 & 4 & 5 & 6 & 7 & 8 & 9 & 10 & 11 & 12 & 13 & 14 & 15   \\
		\hline 
		2 &  \shortstack{6 \\ \tiny{$H_2$}} & \shortstack{4 \\ \tiny{$W$}} & \shortstack{10 \\ \tiny{$H_2$}} & \shortstack{6 \\ \tiny{$W$}} & \shortstack{14 \\ \tiny{$H_2$}} & \shortstack{8 \\ \tiny{$W$}}  & \shortstack{18 \\ \tiny{$H_2$}} & \shortstack{10 \\ \tiny{$W$}} & \shortstack{22 \\ \tiny{$H_2$}} & \shortstack{12 \\ \tiny{$W$}} & \shortstack{26 \\ \tiny{$H_2$}} & \shortstack{14 \\ \tiny{$W$}} & \shortstack{30 \\ \tiny{$H_2$}} \\
		\hline 
		3 & & \shortstack{4 \\ \tiny{$W$}} & \shortstack{20 \\ \tiny{$W_{4,3}$}} & \shortstack{8 \\ \tiny{$H_4$}} & \shortstack{28 \\ \tiny{$W_{4,3}$}} & \shortstack{8 \\ \tiny{$W$}} &  \shortstack{12 \\ \tiny{$H_4$}} & \shortstack{20* \\ \tiny{$W_{4,3}$}} & \shortstack{44 \\ \tiny{$W_{4,3}$}} & \shortstack{12* \\ \tiny{$W_{4,3}$}} & \shortstack{52 \\ \tiny{$W_{4,3}$}} & \shortstack{28* \\ \tiny{$W_{4,3}$}} & \shortstack{20 \\ \tiny{$H_4$}} \\ 
		\hline 
		4 & & & \shortstack{20 \\ \tiny{$H_4$}} & \shortstack{6 \\ \tiny{$W$}} & \shortstack{7 \\ \tiny{$W$}}  & \shortstack{8 \\ \tiny{$W$}} & \shortstack{18* \\ \tiny{$W_{6,4}$}} & \shortstack{10 \\ \tiny{$W$}} & \shortstack{11 \\ \tiny{$W$}} & \shortstack{12 \\ \tiny{$W$}} & \shortstack{13 \\ \tiny{$W$}} & \shortstack{14 \\ \tiny{$W$}} & \shortstack{15 \\ \tiny{$W$}} \\ 
		\hline 
		5 & & & & \shortstack{6 \\ \tiny{$W$}} & \shortstack{42* \\ \tiny{$W_{6,5}$}} & \shortstack{8 \\ \tiny{$W$}} & \shortstack{18* \\ \tiny{$W_{6,5}$}} & \shortstack{10 \\ \tiny{$W$}} & \shortstack{66* \\ \tiny{$W_{6,5}$}} & \shortstack{12 \\ \tiny{$W$}} & \shortstack{78* \\ \tiny{$W_{6,5}$}} & \shortstack{14 \\ \tiny{$W$}} & \shortstack{24 \\ \tiny{$H_8$}} \\ 
		\hline 
		6 & & & & & \shortstack{56 \\ \tiny{$H_8$}} & \shortstack{8 \\ \tiny{$W$}} & \shortstack{24 \\ \tiny{$H_8$}} & \shortstack{40 \\ \tiny{$W_{8,6}$}} & \shortstack{88 \\ \tiny{$W_{8,6}$}} & \shortstack{16 \\ \tiny{$H_8$}}  & \shortstack{104 \\ \tiny{$W_{8,6}$}}  & \shortstack{56 \\ \tiny{$W_{8,6}$}} & \shortstack{40 \\ \tiny{$H_8$}} \\
		\hline 
	\end{tabular}
	\label{tb1}
\end{table}

\begin{remark}
	Table \ref{tb1} gives an overview of achieved minimum number of choice sets ($N$) required to construct an optimal design for given $n$ and $\rho$. This table also helps researchers to decide strategies to run an experiment. Consider the case for $n = 11$ and $\rho = 5$. It is seen from the table that one need 66 choice sets to construct the optimal design. Now the researcher has three options available to reduce the number of choice sets. First, he/she can lower the profile strength $\rho$ to 4, and gets a design in 11 choice sets. Second, he/she can delete one factor that he/she may think is not important and gets a design in 10 choice sets. Third, he/she can add an auxiliary factor in the experiment and gets a design in 12 choice sets.
\end{remark}
\begin{remark}
	Table \ref{tb1} also tells about the construction of optimal choice designs for given  $n$ and $\rho$. If  there is only $W$ written under $N$, then an optimal design can be constructed using weighing matrix $W(n=N, \rho)$. If $W_{\nu,\rho}$ is written under $N$, then an optimal design  can be constructed by $W(\nu,\rho)$ using {\it Method-W} and if $H_r$  is written below $N$, then an optimal design  can be constructed by $H_r$ using {\it Method-H}. A ` $*$' sign after $N$ in some cells convey the fact that {\it Method-W} gives an improved design than {\it Method-H} for those cases. 
\end{remark} 

For $2\leq \rho \leq 6$, $2 < n \leq 15$, Table \ref{tb2} provides the number of choice sets required for achieving optimal designs through {\it Method-W} vis-\`a-vis {\it Method-H}. For the 8 cases, the reduction in number of choice sets through {\it Method-W} over {\it Method-H} ranges from 25\% to 75\%.
\vspace {.2 cm}
\begin{table}[ht]
	\caption{Improved cases using {\it Method-W}}
	\centering
	\begin{tabular}{|c|c|c|c|c|c|c|c|c|}
		\hline
		$(\rho, n)$ & (3, 10) & (3, 12) & (3, 14) & (4, 9) & (5, 7) & (5, 9) & (5, 11) & (5, 13) \\ \hline 
		{\it Method-H} & 40 & 16 & 56 & 36 & 56 & 72 & 88 & 104 \\ \hline
		{\it Method -W} & 20 & 12 & 28 & 18 & 42 & 18 & 66 & 78 \\ \hline 
	\end{tabular}
	\label{tb2}
\end{table}	\vspace{.2cm} 
		
Using the generator technique, we now provide constructions of optimal designs for general choice sets of size $m$. We generate an optimal design $d = (A_1, A_2,\ldots, A_m)$ in $\mathcal{D}_{N,n, m,\rho}$ from an optimal paired design $d' = (A_1, A_2)$ in $\mathcal{D}_{N,n, 2,\rho}$. Let $g_j(w)$ be $j$-th generator with weight $w$, i.e., $g_j(w) = (g_{j1},g_{j2},\ldots, g_{jn})$, where, $g_{jr} = 0,1$, and $\sum_{r=1}^{n} g_{jr} = w$. If $A_i$ is generated from $A_1$ using $g_j(w)$, then we can write $ A_i = A_1 + g_j(w)$,
where, \vspace{-.2cm}
$$ A_i(p,r) = \left\{ \begin{array}{l} A_1(p,r) + g_{jr} \hspace {0.2cm} (\text{mod 2}), \hspace{.2 cm}\text{if $f_r$ is an active factor in $A_1(p,r)$} \vspace{-.2 cm} \\ 
A_1(p,r), \hspace{.2 cm} \text{ if $A_1(p,r)=*$, i.e., $f_r$ is a non-active factor in $A_1(p,r)$,}
\end{array} \right. \vspace{-.2 cm} $$ 
$ 1 \leq p \leq N , \hspace{.2cm}  1 \leq r \leq n$.

\begin{theorem}\label{OptDsn_m}
	Let $G = \{g_1(w_1), g_2(w_2), \ldots,  g_{\alpha}(w_{\alpha})\}$ be the set of $\alpha$ different generators such that both $g_i$ and $\bar{g}_i \notin G$ and $min(\rho, n-\rho) < w_i < max(\rho, n-\rho)$, $i = 1,2, \ldots, \alpha$. If there exists an optimal design  $d'$  in $\mathcal{D}_{N,n,2,\rho}$, then there exists an optimal design  $d$ in $\mathcal{D}_{N,n,m,\rho}$, $m = 1,\ldots,2\alpha + 1,2\alpha + 2 $.
\end{theorem}

\begin{proof}
	Let $d' =(A_1, A_2)$ be an optimal design in $\mathcal{D}_{N,n,2,\rho}$ and let $d = (A_1, A_2, \ldots, A_m)$, where,
	$$ A_{2u+1} = A_1 + g_u(w_u) , \hspace{0.4cm} A_{2u+2} = A_2 + g_u(w_u), \hspace{0.4cm} u = 1, 2,\ldots, \alpha.$$
	
	Take any component paired design $\delta_{ij} = (A_i, A_j)$, $i<j$, $i,j = 1,2,\ldots,m$. For any two columns $h$ and $k$, let $(\eta^{1\pm}_{hk})_{ij}$ is the component part of $\eta^{1\pm}_{hk}$ in the design $\delta_{ij}$. Let $U$ be the set of all indices of columns which are changing between $(A_i,A_j)$. Since $(A_1, A_2)$ is optimal in $\mathcal{D}_{N,n,2,\rho}$, then $(\eta^{1+}_{hk})_{12} = (\eta^{1-}_{hk})_{12} = t$ (say). Therefore, for any $\delta_{ij}$, we have the following two cases
	
	\noindent  {\bf Case 1:} if $h,k \in U$, then 
	$(\eta^{1+}_{hk})_{ij} = (\eta^{1-}_{hk})_{ij} = t$, \\
	\noindent 	{ \bf Case 2:} else,  
	$(\eta^{1+}_{hk})_{ij} = (\eta^{1-}_{hk})_{ij} = 0$.
	
	Considering all possible component paired design $\delta_{ij}$ of $d$, we see that $\eta^{1+}_{hk} = \eta^{1-}_{hk}$, for all $h\neq k, h,k \in \{1,2,\ldots,n\}$. Also note that for every active factor $f_h$ of each choice set $S_{pm}$, $n_{ph} = m/2$ (for even $m$) and $n_{ph} = (m-1)/2$ or $(m+1)/2$ (for odd $m$). Hence, $d$ is optimal in $\mathcal{D}_{N,n,m,\rho}$.      
\end{proof}
Let for a design $d=(A_1,A_2,\ldots,A_m)$, $\bar{d}=(\bar{A}_1, \bar{A}_2,\ldots, \bar{A}_m)$ denotes the complement design of $d$, where $\bar{A}$ is the complement  of $A$ (i.e., 0 and 1 interchange their respective positions in $A$). Then we have the following important theorem.
\begin{theorem}\label{OptDsn_m2}
	If $d$ is optimal in $\mathcal{D}_{N,n,m,\rho}$, then $\bar{d}$ is also optimal in  $\mathcal{D}_{N,n,m,\rho}$. 
\end{theorem}

\begin{proof}
	Since $d$ is optimal in $\mathcal{D}_{N,n,m,\rho}$,  it satisfies both the optimality criteria. Note also that both the optimality criteria are satisfied even if 0 and 1 are interchanged their corresponding positions in $d$.  Hence $\bar{d}$ is also optimal in $\mathcal{D}_{N,n,m,\rho}$.  
\end{proof}

\begin{corollary}\label{corrDsn} 
	If $d$ is optimal in $\mathcal{D}_{N,n,m,\rho}$, then $ \mathbf{d} = \left( \begin{array}{c} d \vspace{-.3cm}\\ \bar{d} \end{array} \right)$ is optimal in  $\mathcal{D}_{2N,n,m,\rho}$.
\end{corollary}
\vspace{-.3cm}
\begin{exmp}\label{ex_m}
	Suppose we want an optimal design for $2^8$ choice experiment with $m = 5$ and $\rho = 6$. Consider the design
	$ d_5 = (A_1, A_2, A_3, A_4, A_5)$, where,  \vspace{-.3cm} 
	\begin{center}
		$  d_5 = \begin{array}{ccccc}
		(111111\!*\!*, & 000000\!*\!*, & 000111\!*\!*, & 111000\!*\!*, & 110000\!*\!*) \vspace{-.3cm} \\
		(101010\!*\!*, & 010101\!*\!*, & 010010\!*\!*, & 101101\!*\!*, & 100101\!*\!*) \vspace{-.3cm} \\
		(1100\!*\!*11, & 0011\!*\!*00, & 0010\!*\!*11, & 1101\!*\!*00, & 1111\!*\!*11) \vspace{-.3cm} \\
		(1001\!*\!*10, & 0110\!*\!*01, & 0111\!*\!*10, & 1000\!*\!*01, & 1010\!*\!*10) \vspace{-.3cm} \\
		(00\!*\!*1111, & 11\!*\!*0000, & 11\!*\!*1111, & 00\!*\!*0000, & 00\!*\!*0011) \vspace{-.3cm} \\
		(01\!*\!*1010, & 10\!*\!*0101, & 10\!*\!*1010, & 01\!*\!*0101, & 01\!*\!*0110) \vspace{-.3cm} \\
		(*\!*\!001100, & *\!*\!110011, & *\!*\!101100, & *\!*\!010011, & *\!*\!110000) \vspace{-.3cm} \\
		(*\!*\!011001, & *\!*\!100110, & *\!*\!111001, & *\!*\!000110, & *\!*\!100101) \\
		\end{array} $.
	\end{center} \vspace{-.3cm}
	Note that the paired design $(A_1,A_2)$ is optimal in $\mathcal{D}_{8,8,2,6}$ and $A_3, A_4, A_5$ is generated from $(A_1,A_2)$ using generators  $g_1 = (11100000)$ and $g_2 = (00111100)$, where  $A_3 = A_1 + g_1$, $A_4 = A_2 + g_1$ and $A_5 = A_1 + g_2$. Thus from Theorem \ref{OptDsn_m}, $d_5$ is optimal in  $\mathcal{D}_{8,8,5,6}$.
	
	Note that if we take a generator $g$ of weight  $ w = 2$, outside the range of $w$, say $g = g_1 = (11000000)$, then the last two choice sets of $d_5$ would look like \vspace{-.5cm}
	\begin{center}
		$ \begin{array}{ccccc}
		(*\!*\!001100, & *\!*\!110011, & *\!*\!001100, & *\!*\!110011,  & *\!*\!110000) \vspace{-.3cm} \\
		(*\!*\!011001, & *\!*\!100110, & *\!*\!011001, & *\!*\!100110,  & *\!*\!100101). \\
		\end{array} $ 
	\end{center} \vspace{-.3cm}
	Here the third and forth treatments are mere repetition of first two treatments. Similar situation occurs if $g$ and $\bar{g}$ are both in $G$. That is why we need both the conditions on $g$ in the Theorem \ref{OptDsn_m}, so that we get choice sets consisting of distinct options.        
\end{exmp}
We end this section with another important construction for general $m$.
\begin{theorem}\label{Multi}
	If there exists an optimal design $d$ in $\mathcal{D}_{N,n,m,\rho}$, then there exists an optimal design $d^t$ in $ \mathcal{D}_{Nt,nt,m,\rho}$, for all $t \geq 1$.    
\end{theorem}

\begin{proof}
	Let $I_t$ be the identity matrix of order $t$ and $d = (A_1, A_2, \ldots, A_m)$ be an optimal design in $\mathcal{D}_{N, n,m,\rho}$.
	Now consider, $d^t = (A_1^t, A_2^t, \ldots, A_m^t)$, where,	$ A_i^t = I_t \otimes A_i,$
	$i = 1,2,\ldots,m$. It is easy to see that $d^t$ satisfies both the criteria of universal optimality in $ \mathcal{D}_{Nt,nt,m,\rho}$. Hence, $d^t$ is optimal in  $\mathcal{D}_{Nt, nt, m, \rho}$.
\end{proof}

\section{ Optimal designs under the broader main effects model}
In many choice situations, the presence of two factor interactions can't be ignored, however researchers may be interested in estimating only the main effects. Under the broader main effects model, choice designs are designs which ensure estimability of all the main effects under the absence of three and higher order interaction effects. Following \cite{r6}, the information matrix of $\theta_1$ corresponding to a design $d \in \mathcal{D}_{N,n,m,\rho}$, under the broader main effects model is 
\begin{equation*}
C = (1/2^n) \{ B_{(1)} \Lambda B'_{(1)} - B_{(1)} \Lambda B'_{(2)}[B_{2} \Lambda B'_{(2)}]^{-}B_{(2)} \Lambda B'_{(1)}\}.
\end{equation*} 

Note that $B_{(1)} \Lambda B'_{(2)}[B_{(2)} \Lambda B'_{(2)}]^{-}B_{(2)} \Lambda B'_{(1)}$ is a non-negative definite matrix and 
$$trace(2^n C) = trace(B_{(1)} \Lambda B'_{(1)}) - trace(B_{(1)} \Lambda B'_{(2)}[B_{(2)} \Lambda B'_{(2)}]^{-}B_{(2)} \Lambda B'_{(1)}).$$
Thus, under the broader main effects model, $trace(C) \leq trace \left((1/2^n)B_{(1)} \Lambda B'_{(1)}\right)$ with equality attaining when $B_{(1)} \Lambda B'_{(2)}$ is a null matrix. For the purpose of this section, we need to define some more notations here. Let $\eta_{hk}^{2+}$ and $\eta_{hk}^{2-}$ are the total number of  component pairs of the type $(01,11)_{(hk)}$ and $(00,10)_{(hk)}$ respectively in  $d$. Similarly, let $\eta_{h(kl)}^{3+}$ and  $\eta_{h(kl)}^{3-}$ are the total number of component pairs of the type  $\{0(a_1a_2), 1(a'_1a'_2)\}_{h(kl)}$, $a_1\neq a_2$, $a'_1 = a'_2$   
and $\{0(a_1a_2), 1(a'_1a'_2)\}_{h(kl)}$, $a_1=a_2$, $a'_1 \neq a'_2$; $a_i, a'_i \in \{0,1\}$, $i =1,2$,  respectively in  $d$. Note also that if $B_{kl}$, $k < l, k,l = 1,\ldots,n$, corresponds to the $r$-th contrast vector in $B_{(2)}$, then $r= \sum_{i = 1}^{k-1} (n-i) + (l-k)$. We now find the necessary and sufficient conditions for  $B_{(1)} \Lambda B'_{(2)}$ to be a null matrix.

\begin{lemma}\label{lem2}
	$B_{(1)} \Lambda B'_{(2)}$ is null if and only if $i) \eta_{hk}^{2+}$ = $\eta_{hk}^{2-}$ and $ii) \eta_{h(kl)}^{3+}$ = $\eta_{h(kl)}^{3-}$, for all $ h\neq k\neq l, h, k, l \in \{1,\ldots , n\}$.
\end{lemma}

\begin{proof}
	Note that $B_{(1)} \Lambda B'_{(2)} = (1/Nm^2) B_{(1)}\Lambda^* B'_{(2)}$. Let $c''_{hr}$ denote the $(h,r)$-th element of $ B_{(1)}\Lambda^* B'_{(2)}$. Then from (\ref{BLB2}) and the {\it Case 2}, {\it Case 3} of Lemma \ref{lemma:converse-3},  we have the following two cases
	\begin{itemize}
		\item[Case 1]
		Let $h$ corresponds to the main effect $F_h$ and $r$ corresponds to the two factor interaction effect $F_{hk}$. Then,
		\begin{eqnarray*}\label{Cstarxx}
			c''_{hr} &=& \sum_{j_{1}<j_{2}< \cdots <j_{m}} N_{j_1 j_2 \ldots j_m}\sum_{j_{r}<j_{r'}}\{B_hM^{(j_{r}j_{r'})} B_{hk}'\}\\
			&=& \left[4(\eta_{hk}^{2+} -\eta_{hk}^{2-})+ 0\{N^* - (\eta_{hk}^{2+} + \eta_{hk}^{2-})\} \right].
		\end{eqnarray*}
		Thus $c''_{hr}=0$ if and only if $\eta_{hk}^{2+} = \eta_{hk}^{2-}$.
		
		\item[Case 2]  Let $h$ corresponds to the main effect $F_h$ and $r$ corresponds to the two factor interaction effect $F_{kl}$. Then,
		\begin{eqnarray*}\label{Cstarxx}
			c''_{hr} &=& \sum_{j_{1}<j_{2}< \cdots <j_{m}} N_{j_1 j_2 \ldots j_m}\sum_{j_{r}<j_{r'}}\{B_hM^{(j_{r}j_{r'})} B_{kl}'\}\\
			&=& \left[ 4(\eta_{h(kl)}^{3+} - \eta_{h(kl)}^{3-})+ 0\{N^* - (\eta_{h(kl)}^{3+} + \eta_{h(kl)}^{3-})\}\right].
		\end{eqnarray*}
	\end{itemize}
	Thus $c''_{hr}=0$ if and only if $\eta_{h(kl)}^{3+} = \eta_{h(kl)}^{3-}$.
\end{proof}

In what follows, $d$ is optimal in  $\mathcal{D}_{N,n,m, \rho}$ under the broader main effects model if (a) $d$ is optimal under the main effects model and (b) $B_{(1)} \Lambda B'_{(2)}$ is a null matrix i.e.
i) $\eta_{hk}^{2+} = \eta_{hk}^{2-}$,  ii)  $\eta_{h(kl)}^{3+} = \eta_{h(kl)}^{3-}$, 
for all $ h\neq k\neq l, h,k,l \in \{1,2,\ldots,n\} $.\\

The following theorem provides a construction for optimal designs under the broader main effects model. 
\begin{theorem}
	If $d$ is an optimal design in $\mathcal{D}_{N,n,m,\rho}$ under the main effects model, then $ \mathbf{d} = \left( \begin{array}{c} d \vspace{-.3cm} \\ \bar{d} \end{array} \right)$ is optimal in  $\mathcal{D}_{2N,n,m,\rho}$ under the broader main effects model.  
\end{theorem}

\begin{proof}
	Note from Corollary \ref{corrDsn} that $\mathbf{d}$ is optimal in  $\mathcal{D}_{2N,n,\rho, m }$ under the main effects  model. Therefore, we only need to show that $\mathbf{d}$ also satisfies Lemma \ref{lem2}. Consider a component paired design $\delta^*_{ij} =  \left( \begin{array}{c} \delta_{ij} \vspace{-.3cm} \\ \bar{\delta}_{ij} \end{array} \right)$, $i < j$, $i,j = 1,2, \ldots, m$, of $\mathbf{d}$. Note that for any two columns $h$ and $k$, if there exists a pair of type $(00,10)_{(hk)}$  in $\delta_{ij}$, then the type $(01,11)_{(hk)}$ exists in the corresponding pair of $\bar{\delta}_{ij}$ and vice-versa. Similarly, for any three columns $h$, $k$ and $l$, if there exists a  pair of type $\{0(a_1a_2), 1(a'_1a'_2)\}_{h(kl)}$, $a_1=a_2$, $a'_1 \neq a'_2$,  in $\delta_{ij}$, then the type $\{0(a_1a_2), 1(a'_1a'_2)\}_{h(kl)}$, $a_1\neq a_2$, $a'_1 = a'_2$; $a_i, a'_i \in \{0,1\}$, $i =1,2$, exists in the corresponding pair of $\bar{\delta}_{ij}$ and vice-versa. Thus considering all possible $m(m-1)/2$ component paired designs $\delta^*_{ij}$ in $\mathbf{d}$, we see that $\mathbf{d}$ satisfies  Lemma \ref{lem2}. Hence the theorem.    
\end{proof}

\begin{exmp}
	From Example \ref{ex_m} we see that $d_5$ is optimal in $\mathcal{D}_{8,8,5,6}$. Then $ \mathbf{d}_5 = \left( \begin{array}{c} d_5 \vspace{-.3cm} \\ \bar{d}_5 \end{array} \right)$ is optimal in $\mathcal{D}_{16,8,5,6}$ under the broader main effects model. 
\end{exmp}
\begin{theorem}
	Under the broader main effects model set-up, if there exists an optimal design $d$ in $\mathcal{D}_{N,n,m,\rho}$, then there exists an optimal design $d^t$ in $ \mathcal{D}_{Nt,nt,m,\rho}$, for all $t \geq 1$.    
\end{theorem}
\begin{proof}
	The proof is similar to the proof of Theorem \ref{Multi}.
\end{proof}
\vspace{.5cm}

\noindent {\bf Acknowledgment:} The authors thank Prof. Niranjan Balachandran and Ms. Rakhi Singh of IIT Bombay for their constructive comments on this article, which have substantially improved the clarity of its presentation.  
\vspace{.3cm} \\

\end{document}